%% file: main.tex
\documentclass[11pt]{llncs}
\usepackage{graphicx}
\usepackage{algorithm}
\usepackage{algorithmic}
\usepackage{amsmath}
\usepackage{amssymb}
\usepackage{amsfonts}
\usepackage{color}
\usepackage[colorlinks=true, linkcolor=blue, citecolor=blue, urlcolor=blue]{hyperref}
\usepackage{bm}
\usepackage{xcolor}
\usepackage{caption}
\usepackage{multirow}
\usepackage{diagbox}
\usepackage{array,makecell,booktabs}
\usepackage{threeparttable}
\newcolumntype{M}[1]{>{\centering\arraybackslash}m{#1}}
\usepackage{adjustbox}
\usepackage{listings}
\usepackage{matlab-prettifier}
\usepackage{subfigure}
\usepackage{bbding}
\usepackage{footmisc}
\usepackage{cleveref}
\usepackage{longtable}
\usepackage[misc]{ifsym}
\usepackage{ulem}
\usepackage{orcidlink}
\usepackage[T1]{fontenc}
\usepackage{todonotes}

\newcommand{\oomit}[1]{}

\lstdefinestyle{commonstyle}{
  basicstyle=\ttfamily\small,
  breaklines=true,
  prebreak=\raisebox{0ex}[0ex][0ex]{\ensuremath{\hookleftarrow}},
  frame=tb,
  numbers=left,
  xleftmargin=0.6cm,
  numberstyle=\tiny\color{gray},
  showstringspaces=false,
  keywordstyle=\bfseries\color{blue},
  commentstyle=\color{green!40!black},
  stringstyle=\color{purple},
  rulecolor=\color{black},
  backgroundcolor=\color{white},
  tabsize=4,
  columns=fullflexible,
  label=\lstname
}

\lstdefinestyle{pythonstyle}{
  language=Python,
  style=commonstyle
}

\begin{document}

\title{Quantitative Verification of Constrained Occupation Time for Stochastic Discrete-time Systems
}
\titlerunning{Quantitative Verification of Constrained Occupation Time}

\author{Bai Xue$^1$, Peixin Wang$^2$, and C.-H. Luke Ong$^3$} 

\institute{1. KLSS, Institute of Software, Chinese Academy of Sciences, Beijing, China\\
 Email: xuebai@ios.ac.cn\\
 2. Software Engineering Institute, East China Normal University, China \\ Email: pxwang@sei.ecnu.edu.cn\\
3. College of Computing and Data Science, Nanyang Technological University, Singapore\\
Email: luke.ong@ntu.edu.sg
}

\maketitle              

\begin{abstract}
This paper addresses the quantitative verification of constrained occupation time in stochastic discrete-time systems, focusing on the probability of visiting a target set at least $k$ times while maintaining safety. Such cumulative properties are essential for certifying repeated behaviors like surveillance and periodic charging. To address this, we present the first barrier certificate framework capable of certifying these behaviors. We introduce multiplicative stochastic barrier functions that encode visitation counts implicitly within the algebraic structure of a scalar barrier. By adopting a switched-system reformulation to handle safety, we derive rigorous probabilistic bounds for both finite and infinite horizons. Specifically, we show that dissipative barriers establish upper bounds ensuring the exponential decay of frequent visits, while attractive barriers provide lower bounds via submartingale analysis. The efficacy of the proposed framework is demonstrated through numerical examples. 
\end{abstract}

\input{introduction}

\input{preliminaries}

\input{conditions_new}
\input{examples}

\input{conclusion}

 \bibliographystyle{splncs04}
 \bibliography{ref}
 \input{appendix}
\end{document}

%% file: introduction.tex
\section{Introduction}
\label{sec:intr}
Stochastic discrete-time systems provide a canonical model for safety-critical and decision-making processes operating under uncertainty, including autonomous robots \cite{soderstrom2012discrete}, cyber-physical systems \cite{clarke2011statistical}, and probabilistic programs \cite{gordon2014probabilistic}. Formal verification of such systems aims to provide mathematically rigorous guarantees that desired specifications hold with quantifiable probability \cite{clarke1997model,baier2008principles,manna2012temporal}. Over the past decades, substantial progress has been made in the verification of single-event specifications, such as safety (invariance), reachability, and reach-avoid properties, using techniques ranging from dynamic programming and Bellman equations (e.g., \cite{abate2008probabilistic,summers2010verification}) to  barrier certificates (e.g., \cite{prajna2007framework,chakarov2013probabilistic,henzinger2025supermartingale}).

However, many realistic specifications are inherently cumulative rather than event-based, necessitating the consideration of occupation times \cite{darling1957occupation,godreche2001statistics}. In surveillance, inspection, maintenance, or energy-aware autonomy, correctness is not determined by whether a target region is reached once, but by how frequently the system revisits certain regions while remaining safe. For example, consider a drone monitoring a wildfire. It is not sufficient for the drone to simply reach the monitoring zone once; rather, it must remain in that zone for a cumulative duration (e.g., 10 minutes) to gather sufficient data, all while staying within a safe flight envelope. Similarly, consider an autonomous underwater vehicle (AUV) inspecting a submerged pipeline. Due to strong and stochastic ocean currents, stationary hovering is infeasible. Instead, the AUV must execute repeated passes or dynamic station-keeping maneuvers to keep its sensors focused on the target area. The mission requires accumulating a prescribed amount of valid scan data—corresponding to occupation time—while counteracting drift forces that may push the vehicle toward hazardous obstacles. In such scenarios, the appropriate metric of success is the constrained occupation time: the total time spent in a target set before any safety violation occurs. Despite their practical relevance, such specifications remain largely underexplored in stochastic verification. 

In this paper, we propose a novel verification paradigm for constrained occupation-time properties. By introducing a multiplicative scaling mechanism that amplifies or attenuates the barrier value upon each target visit, we derive rigorous probabilistic upper and lower bounds on repeated visitation events for both finite and infinite horizons. To handle safety constraints, we adopt a switched-system reformulation that freezes the dynamics upon safety violation. This transformation allows the constrained occupation-time problem to be reduced to an unconstrained occupation-time analysis of an auxiliary stochastic process, enabling the systematic application of martingale-based techniques. The resulting framework yields two complementary classes of certificates. Dissipative stochastic barriers establish exponential upper bounds on the probability of frequent target visitation. Attractive stochastic barriers, on the other hand, provide quantitative lower bounds, certifying repeated interaction with the target region. Together, these results extend classical reach-avoid verification to a new class of cumulative specifications that cannot be efficiently addressed by existing barrier-based methods. 
Finally, the effectiveness of the proposed barrier conditions is demonstrated through two illustrative examples.

Our contributions are summarized below:
\begin{enumerate}
    \item \textbf{A novel certificate-based formulation of constrained occupation-time verification.}
We introduce, for the first time, stochastic barrier functions with a multiplicative structure that implicitly encodes visitation counts for constrained occupation-time verification. 
\item \textbf{Two-sided quantitative guarantees.}
We derive sufficient conditions for both upper and lower bounds on constrained occupation probabilities. Dissipative barriers certify exponential decay of frequent visitation, while attractive barriers certify repeated intersections under safety constraints.
\end{enumerate}

\subsection*{Related Work}

The verification of stochastic systems has evolved significantly over the past two decades, moving from qualitative almost-sure guarantees to quantitative probabilistic bounds \cite{kushner1967stochastic,prajna2007framework,chakarov2013probabilistic,kenyon2021supermartingales,takisaka2021ranking,yu2023safe}. This section reviews closely related barrier-based methods for the formal verification of stochastic discrete-time systems, categorizing them into infinite-horizon and finite-horizon approaches.

\noindent\textbf{Infinite-Horizon Verification.}
Barrier certificates, originally inspired by Lyapunov functions for stability analysis, were introduced to stochastic (hybrid) systems in \cite{prajna2007framework} to certify safety (invariance) properties using supermartingales. This foundational work used Ville's inequality to upper-bound the probability of exiting a safe set over an infinite horizon.  Subsequently, research expanded to reach-avoid specifications—reaching a target while remaining safe \cite{prajna2007convex,cao2025comparative}.
\begin{itemize}
    \item Ranking Supermartingales: Early approaches combined barrier functions with ranking supermartingales to certify that a system eventually reaches a target set with probability one (almost-sure reachability) or to bound termination probabilities in probabilistic programs \cite{chakarov2013probabilistic,chatterjee2017stochastic}. 

     \item Reach-Avoid Supermartingales: Under the assumption of  the existence of a robust invariant set (a domain the system never leaves), recent works have introduced specialized certificate structures such as Additive and Multiplicative Reach-Avoid Supermartingales \cite{vzikelic2023compositional}. 
    
    \item Relaxed Bellman-like Equations: A more direct approach involves relaxing the Bellman-like equations associated with the reachability probability \cite{xue2021reach,yu2023safe}. Recently, \cite{xue2024sufficient} proposed necessary and sufficient barrier conditions for safety and reach-avoid verification. These methods can provide both lower and upper bounds by reversing inequalities.
\end{itemize}

\noindent\textbf{Finite-Horizon Verification.}
Finite-time verification is also attracting much attention for physical systems with operational deadlines.
\begin{itemize}
    \item c-Martingales: Building on foundational work \cite{kushner1967stochastic}, $c$-martingales were introduced to bound the failure probability over a finite time horizon \cite{steinhardt2012finite}. This framework was subsequently extended to finite-time temporal logic verification \cite{pnueli1977temporal} under state invariance constraints \cite{jagtap2018temporal}, to stopped stochastic processes \cite{santoyo2021barrier}, and later adapted to discrete-time systems and neural barrier certificates \cite{mathiesen2022safety}. These works primarily focus on deriving upper bounds on failure probabilities. More recently, the $c$-martingale framework was generalized to provide both lower and upper bounds on safety violations \cite{zhi2024unifying}.
    
    \item Switched System Approaches: To mitigate the reliance on robust invariants, \cite{xue2024finite} proposed a switched-system formulation. By freezing the dynamics upon safety violation, this approach allows for the computation of both lower and upper bounds for safety and reach-avoid probabilities without requiring the system to stay within a robust invariant.
\end{itemize}
The aforementioned works address single-event properties (reaching a target or unsafe set once). In contrast, this work addresses the problem of  constrained occupation time—visiting a target $k \ge 1$ times while remaining safe in the context of stochastic barriers. 

The verification of $\omega$-regular properties represents a significant generalization of the reach-avoid problem. These properties typically involve infinite-time behaviors, such as recurrence (visiting a set infinitely often, i.e., the Büchi condition) or persistence (eventually staying within a set) \cite{chakarov2016deductive,abate2024stochastic,abate2025quantitative,henzinger2025supermartingale}. Recently, \cite{wang2025quantitative} addresses quantitative verification of $\omega$-regular properties in probabilistic programs by decomposing $\omega$-regular properties (or, the Rabin acceptance condition) into persistence and recurrence components. \textit{While their analysis of recurrence bounds the probability of visiting a set at most $k$ times, our work addresses the constrained  occupation time problem, which explicitly requires visiting a target at least $k$ times while strictly avoiding unsafe regions—a safety constraint not directly targeted by their recurrence barriers.} On the other hand, their reliance on augmenting the state space with explicit counter variables introduces a dependency where the system dimension grows with the visit requirement $k$. 
In distinct contrast, our framework avoids state augmentation entirely. We introduce a novel mechanism to encode the visitation count implicitly into the multiplicative structure of a scalar barrier field defined solely on the original state space. This decouples the complexity of the verification condition from the magnitude of $k$, offering a theoretically superior representation for high-frequency occupation tasks.

%% file: preliminaries.tex
The paper is structured as follows: Section \ref{sec:preliminaries} outlines the problem setting, while Section \ref{sec:qvot} details our barrier methods for quantitative verification of constrained occupation times. Section \ref{sec:examples} demonstrates the efficacy of the proposed barrier methods using two scalar systems. Finally, Section \ref{sec:con} concludes.

\noindent\textbf{Notations}: Let $\mathbb{R}$, $\mathbb{R}_{\geq 0}$, $\mathbb{N}$, and $\mathbb{N}_{\le k}$ denote the sets of real numbers, non-negative real numbers, non-negative integers, and non-negative integers up to $k$, respectively. For a set $\mathcal{A}$, $\mathbf{1}_{\mathcal{A}}(x)$ denotes the indicator function, which equals 1 if $x \in \mathcal{A}$ and 0 otherwise. Similarly, for a condition $P$, we denote the indicator $\mathbf{1}_{{P}}$ as 1 if $P$ is true and 0 otherwise. We denote the set difference by $A \setminus B$.

\section{Preliminaries}
\label{sec:preliminaries}
This section introduces the mathematical framework used in the paper. We define the stochastic system model and the underlying probability space, introduce constrained occupation time verification problem, and recall key martingale results—namely the Optional Stopping Theorem \cite{williams1991probability} and Markov’s inequality \cite{grimmett2020probability}—that underpin our probabilistic analysis.

\subsection{Problem Formulation}
We consider a discrete-time stochastic dynamical system governed by
\begin{equation}
\label{system}
    X_{t+1} = f(X_t, d_t), \quad t \ge 0,
\end{equation}
where $X_t \in \mathbb{R}^n$ denotes the system state, $\{d_t\}_{t \ge 0}$ is a sequence of independent and identically distributed (i.i.d.) disturbances taking values in a measurable set $\mathcal{D} \subseteq \mathbb{R}^m$, and $f:\mathbb{R}^n \times \mathcal{D}\rightarrow \mathbb{R}^n$.

To rigorously analyze the probabilistic behavior of \eqref{system}, we define the associated probability space explicitly. Let $(\mathcal{D}, \mathcal{F}_{\mathcal{D}}, \mathbb{P}_{\mathcal{D}})$ be the probability space for a single disturbance, where $\mathcal{F}_{\mathcal{D}}$ is the Borel $\sigma$-algebra over $\mathcal{D}$ and $\mathbb{P}_{\mathcal{D}}$ is the probability measure. The canonical sample space is defined as $\Omega = \mathcal{D}^\infty$, representing the set of all infinite disturbance sequences $\omega = (d_0, d_1, \dots)$. We endow $\Omega$ with the product topology and the product probability measure $\mathbb{P} = \mathbb{P}_{\mathcal{D}}^\infty$. The expectation operator with respect to $\mathbb{P}$ is denoted by $\mathbb{E}[\cdot]$. We also define the natural filtration $\mathcal{F}_t = \sigma(X_0, \dots, X_t)$, which captures the information history available up to time $t$.

Given an initial state $X_0 \in \mathbb{R}^n$ and a disturbance realization $\omega \in \Omega$, the system \eqref{system} generates a unique state trajectory. We denote this deterministic sequence by $\phi_{\omega}^{X_0}: \mathbb{N} \to \mathbb{R}^n$, defined recursively as:
\begin{align}
    \phi_{\omega}^{X_0}(0) &= X_0, \\
    \phi_{\omega}^{X_0}(t+1) &= f(\phi_{\omega}^{X_0}(t), d_t(\omega)), \quad \forall t \ge 0,
\end{align}
where $d_t(\omega)$ denotes the $t$-th component of $\omega$. Consequently, the random variable $X_t$ in \eqref{system} is the mapping $\omega \mapsto \phi_{\omega}^{X_0}(t)$.

\noindent\textbf{Constrained Occupation Time and Its Quantitative Verification Problem.}
Let $\mathcal{X} \subseteq \mathbb{R}^n$ be a designated \textit{safe set}. The system is deemed safe as long as the state remains within $\mathcal{X}$. We define the \textit{safety exit time}, $\tau_{\text{safe}}$, as the first time instant the system violates the safety constraint:
\begin{equation}
    \tau_{\text{safe}}(\omega) := \inf \{ t \ge 0 : \phi_{\omega}^{X_0}(t) \notin \mathcal{X} \}.
\end{equation}
By convention, if the trajectory never leaves $\mathcal{X}$, we set $\tau_{\text{safe}} = \infty$.

Consider a measurable \textit{target set} $\mathcal{T} \subseteq \mathcal{X}$ located entirely within the safe region. We aim to quantify the frequency of system interactions with $\mathcal{T}$ prior to any safety violation. To this end, we define the \textit{constrained occupation time}, $N_{\mathcal{T}}(N, \omega)$, as the cumulative number of visits to $\mathcal{T}$ over a horizon $N \in \mathbb{N}_{\ge 1} \cup \{\infty\}$, counting only those visits that occur while the system remains safe:
\begin{equation}
    N_{\mathcal{T}}(N, \omega) := \sum_{t=0}^{N} \mathbf{1}_{\mathcal{T}}(\phi_{\omega}^{X_0}(t)) \cdot \mathbf{1}_{\{t < \tau_{\text{safe}}(\omega)\}},
\end{equation}
where $X_0\in \mathcal{X}$.

Our primary objective is to bound the probability that the constrained occupation time attains at least $k$ visits.
\begin{problem}[Quantitative Verification on Constrained Occupation Time]
\label{pro}
Given a horizon $N$ and a visit threshold $k$ (with $1 \leq k \leq N+1$ for finite $N$, or $k \ge 1$ for $N=\infty$), our objective is to compute rigorous lower and upper bounds $\epsilon_1, \epsilon_2 \in [0,1]$ such that:
\begin{equation*}
    \epsilon_1 \leq \mathbb{P}(N_{\mathcal{T}}(N,\omega) \ge k) \leq \epsilon_2,
\end{equation*}
where the probability is defined over the space of disturbance trajectories $\Omega$.
\end{problem}

\subsection{Martingales, Stopping Theorems, and Markov’s Inequality}
This subsection recalls fundamental results from martingale theory \cite{williams1991probability} - specifically stopping times, the Optional Stopping Theorem, and Markov's inequality - that enable the rigorous analysis of system behavior over random horizons.

\begin{definition}[Stopping Time]
A random variable $\tau: \Omega \to \mathbb{N} \cup \{\infty\}$ is a \textit{stopping time} with respect to the filtration $\{\mathcal{F}_t\}_{t \ge 0}$ if, for every deterministic time $t \ge 0$, the event $\{\tau = t\}$ is $\mathcal{F}_t$-measurable.
\end{definition}

\begin{definition}[Martingale, Supermartingale, Submartingale]
\label{def:martingales}
An adapted process $\{M_t\}_{t \ge 0}$ with $\mathbb{E}[|M_t|] < \infty$ is called a \textit{martingale} with respect to $\mathcal{F}_t$ if for all $t \ge 0$, $\mathbb{E}[M_{t+1} \mid \mathcal{F}_t] = M_t$. The process is called a \textit{Supermartingale} if $\mathbb{E}[M_{t+1} \mid \mathcal{F}_t] \le M_t$, and \textit{Submartingale} if $\mathbb{E}[M_{t+1} \mid \mathcal{F}_t] \ge M_t$.
\end{definition}

To handle analysis over random horizons (such as $\tau_{\text{safe}}$), we utilize two fundamental results from martingale theory.

\begin{proposition}[Optional Stopping Theorem \cite{williams1991probability}]
\label{thm:OST}
Let $\{M_t\}_{t \ge 0}$ be an adapted process and let $\tau$ be a stopping time.
\begin{enumerate}
    \item \textbf{Supermartingale case.}
    If $\{M_t\}$ is a nonnegative supermartingale, then
    $\mathbb{E}\!\left[M_\tau \mathbf{1}_{\{\tau < \infty\}}\right] \le \mathbb{E}[M_0]$.
    In particular, if $\tau$ is almost surely bounded, then $\mathbb{E}[M_\tau] \le \mathbb{E}[M_0]$.
    
    \item \textbf{Submartingale case.}
    If $\{M_t\}$ is a submartingale and $\tau$ is almost surely bounded, then $\mathbb{E}[M_\tau] \ge \mathbb{E}[M_0]$.
\end{enumerate}
\end{proposition}

\begin{proposition}[Markov's Inequality \cite{grimmett2020probability}]
\label{thm:markov}
Let $X$ be a non-negative random variable defined on a probability space $(\Omega, \mathcal{F}, \mathbb{P})$. For any constant $a > 0$, $\mathbb{P}(X \ge a) \le \frac{\mathbb{E}[X]}{a}$ holds.
\end{proposition}

%% file: conditions_new.tex
\section{Quantitative Verification of Constrained Occupation Time}
\label{sec:qvot}

In this section, we propose new stochastic barrier conditions for characterizing upper and lower bounds on the probability defined in Problem \ref{pro}. The sufficient condition for upper bounds is introduced in Subsection \ref{sub:scub1}, and the condition for lower bounds is formulated in Subsection \ref{sub:sclb1}. 

To facilitate this analysis, we adopt the switched-system methodology from \cite{xue2021reach,xue2024finite}. The core idea is to ``freeze'' the dynamics of the original system \eqref{system} the moment it exits the safe set $\mathcal{X}$. This transformation encodes the safety constraint implicitly within the dynamics, effectively reducing the constrained occupation-time problem to an unconstrained analysis of an auxiliary stochastic process.
\begin{definition}
\label{switch}
A \textit{switched stochastic system} is the tuple $(\widetilde{\mathcal{X}}, \widetilde{f}, X_0)$, where:
\begin{enumerate}
    \item $\widetilde{\mathcal{X}} \subseteq \mathbb{R}^n$ is the augmented state space containing the one-step reachable set 
    \[
        \widetilde{\Omega} = \mathcal{X} \cup \{f(X,d) \mid X \in \mathcal{X},\, d \in \mathcal{D}\} \subseteq \widetilde{\mathcal{X}};
    \]
    \item $X_0 \in \mathcal{X}$ is the initial state;
    \item $\widetilde{f} \colon \widetilde{\mathcal{X}} \times \mathcal{D} \to \widetilde{\mathcal{X}}$ is the transition function, with absorbing dynamics outside $\mathcal{X}$:
    \[
        \widetilde{f}(\widetilde{X},d) = 
        \begin{cases} 
            f(\widetilde{X},d), & \widetilde{X} \in \mathcal{X}, \\
            \widetilde{X}, & \widetilde{X} \in \widetilde{\mathcal{X}} \setminus \mathcal{X}.
        \end{cases}
    \]
\end{enumerate}
The state evolves as
\begin{equation}
\label{sdss1}
   \widetilde{X}_{t+1} = \widetilde{f}(\widetilde{X}_t, d_t), \quad t \ge 0. 
\end{equation}
\end{definition}

Let $\widetilde{\bm{\phi}}_{\pi}^{X_0}(\cdot)$ denote the trajectory of the switched system \eqref{sdss1} under a disturbance signal $\pi$. The switched system effectively mirrors the original dynamics until a safety violation occurs. Specifically, if the original system \eqref{system} exits $\mathcal{X}$ at time $\tau_{\text{safe}}$, the switched system enters the region $\widetilde{\mathcal{X}} \setminus \mathcal{X}$ at the same instant. Crucially, the identity dynamics in the second case of $\widetilde{f}$ ensure that $\widetilde{\mathcal{X}} \setminus \mathcal{X}$ acts as an absorbing set. Consequently, the system state remains trapped in $\widetilde{\mathcal{X}} \setminus \mathcal{X}$ for all $t \ge \tau$, making the domain $\widetilde{\mathcal{X}}$ a robust invariant.

We define the occupation time for the switched system over the horizon $N$ (including $t=0$ to $N$) as:
\begin{equation*}
    \widetilde{N}_{\mathcal{T}}(N,\omega) := \sum_{t=0}^{N} \mathbf{1}_{\mathcal{T}}(\widetilde{\bm{\phi}}_{\omega}^{X_0}(t)).
\end{equation*}

The following lemma establishes the equivalence between the constrained occupation problem for the original system and the unconstrained occupation problem for the switched system \eqref{sdss1}.

\begin{lemma}
\label{lema1}
For any horizon $N \in \mathbb{N}$ and threshold $k \in \mathbb{N}$ with $1 \le k \le N+1$,  
\[
\mathbb{P}(N_{\mathcal{T}}(N,\omega) \ge k) = \mathbb{P}(\widetilde{N}_{\mathcal{T}}(N,\omega) \ge k).
\]
\end{lemma}

\begin{proof}
    Let $\omega \in \Omega$ be an arbitrary disturbance signal and let $\tau_{\text{safe}}(\omega)$ be the safety exit time of the original system \eqref{system}. We compare the accumulation of visits to $\mathcal{T}$ for both systems.

    \textbf{1. Before Exit ($t < \tau_{\text{safe}}(\omega)$):}
    By the definition of $\tau_{\text{safe}}$, for all $0 \le t < \tau_{\text{safe}}(\omega)$, the state satisfies $\bm{\phi}_{\omega}^{X_0}(t) \in \mathcal{X}$. Since $\widetilde{f}(\widetilde{X}, d) = f(\widetilde{X}, d)$ for all $\widetilde{X} \in \mathcal{X}$, the trajectory of the switched system coincides with that of the original system:
    \[
    \widetilde{\bm{\phi}}_{\omega}^{X_0}(t) = \bm{\phi}_{\omega}^{X_0}(t), \quad \forall t < \tau_{\text{safe}}(\omega).
    \]
    Consequently, $\mathbf{1}_{\mathcal{T}}(\widetilde{\bm{\phi}}_{\omega}^{X_0}(t)) = \mathbf{1}_{\mathcal{T}}(\bm{\phi}_{\omega}^{X_0}(t))$ for all such $t$.

    \textbf{2. After Exit ($t \ge \tau_{\text{safe}}(\omega)$):}
    If the system exits $\mathcal{X}$, then at time $t = \tau_{\text{safe}}(\omega)$, the state satisfies $\bm{\phi}_{\omega}^{X_0}(t) \in \widetilde{\mathcal{X}} \setminus \mathcal{X}$. By Definition \ref{switch}, the switched system enters this region and remains there for all $t \ge \tau_{\text{safe}}(\omega)$ (i.e., $\widetilde{\bm{\phi}}_{\omega}^{X_0}(t) \in \widetilde{\mathcal{X}} \setminus \mathcal{X}$).
    
    Recall that the target set is entirely within the safe set ($\mathcal{T} \subseteq \mathcal{X}$); thus, it is disjoint from the region $\widetilde{\mathcal{X}} \setminus \mathcal{X}$. This implies
    \[
    \mathbf{1}_{\mathcal{T}}(\widetilde{\bm{\phi}}_{\omega}^{X_0}(t)) = 0, \quad \forall t \ge \tau_{\text{safe}}(\omega).
    \]
    For the original system, the definition of $N_{\mathcal{T}}(N,\omega)$ includes the term $\mathbf{1}_{\{t < \tau_{\text{safe}}\}}$, ensuring that for $t \ge \tau_{\text{safe}}(\omega)$, the contribution to the count is also zero.

    \textbf{Conclusion:}
    Combining these observations, we have
    \begin{align*}
        \widetilde{N}_{\mathcal{T}}(N,\omega) &= \sum_{t=0}^{N} \mathbf{1}_{\mathcal{T}}(\widetilde{\bm{\phi}}_{\omega}^{X_0}(t)) = \sum_{t=0}^{\min(N, \tau_{\text{safe}}(\omega))-1} \mathbf{1}_{\mathcal{T}}(\bm{\phi}_{\omega}^{X_0}(t)) + 0 \\
        &= \sum_{t=0}^{N} \mathbf{1}_{\mathcal{T}}(\bm{\phi}_{\omega}^{X_0}(t)) \cdot \mathbf{1}_{\{t < \tau_{\text{safe}}(\omega)\}} = N_{\mathcal{T}}(N,\omega).
    \end{align*}
    Since the random variables $N_{\mathcal{T}}(N,\omega)$ and $\widetilde{N}_{\mathcal{T}}(N,\omega)$ are identical for every sample path $\omega$, their probability distributions are identical. \qed
\end{proof}

Therefore, in the remainder of this paper, we focus our analysis on bounding $\mathbb{P}(\widetilde{N}_{\mathcal{T}}(N,\omega) \ge k)$.

\subsection{Dissipative Barriers for Upper Bounds}
\label{sub:scub1}

In this subsection, we establish a sufficient condition to upper bound the probability of the constrained occupation time exceeding a threshold $k$. We utilize the switched system introduced in Definition \ref{switch}. To simplify notation, in this subsection $\widetilde{X}_t$ denotes the state of the switched system \eqref{sdss1}.

Theorem \ref{thm:mult_beta_occupation} presents a stochastic barrier condition relying on a dissipative function $v: \widetilde{\mathcal{X}} \to \mathbb{R}_{\ge 0}$ to establish an upper bound on the probability that the switched system \eqref{sdss1} visits $\mathcal{T}$ at least $k$ times within the horizon $N$. By Lemma \ref{lema1}, this corresponds directly to the probability that the original system \eqref{system} visits $\mathcal{T}$ at least $k$ times before leaving the safe set $\mathcal{X}$. This theorem interprets $v(x)$ as a dissipative energy measure with a contraction rate $\alpha \in (0,1)$. The proof relies on a scaled process $Z_t$ that encodes the occupation count by amplifying the state value by $\alpha^{-1}$ upon each visit to $\mathcal{T}$. Although this scaling artificially inflates trajectories with frequent visits, the drift condition ensures that the global expected energy remains controlled. Consequently, for finite horizons, the specific event of realizing $k$ visits—which necessitates a geometric surge in $Z_t$ exceeding $\alpha^{-k}$—is constrained by Markov's inequality to occur with exponentially low probability. In the infinite horizon limit, this guarantee holds strictly provided the additive drift vanishes ($\beta=0$).

\begin{theorem}[Occupation-Time Upper Bounds via Dissipative Barriers]
\label{thm:mult_beta_occupation}
Suppose there exist a measurable function $v:\widetilde{\mathcal{X}} \to \mathbb{R}_{\ge 0}$ and constants $\alpha \in (0,1)$ and $\beta \ge 0$ such that:
\begin{enumerate}
    \item \textbf{Multiplicative Barrier Condition:}
    \begin{equation}
    \label{eq:mult_beta_drift}
    \mathbb{E}\left[ \widetilde{v}(X_{t+1}) \mid X_t \right] \le \alpha v(X_t) + \beta, \quad \forall X_t \in \mathcal{X},
    \end{equation}
    where $
    \widetilde{v}(X) := \mathbf{1}_{\mathcal{T}}(X) v(X) + \alpha \mathbf{1}_{\widetilde{\mathcal{X}} \setminus \mathcal{T}}(X) v(X)$.

    \item \textbf{Positivity on the target:}\quad $    v(x) \ge 1, \quad \forall x \in \mathcal{T}$. 
\end{enumerate}
Then, for all integers $k$ such that $1 \le k \le N+1 < \infty$,
\begin{equation}
\label{eq:mult_beta_bound}
\mathbb{P}\big(\widetilde{N}_{\mathcal{T}}(N,\omega) \ge k \big) \le \left( v(X_0)\rho(X_0) + \frac{\beta}{\alpha} \sum_{t=1}^{N} \alpha^{-t} \right) \alpha^k,
\end{equation}
where $X_0\in \mathcal{X}$, and $\rho(X_0) = \alpha^{-1}$ if $X_0 \in \mathcal{T}$ and $\rho(X_0) = 1$ otherwise.
In the special case where $N = \infty$ and $\beta = 0$, the bound simplifies to:
\begin{equation}
\label{eq:mult_beta_bound_inf}
\mathbb{P}\big(\widetilde{N}_{\mathcal{T}}(\infty,\omega) \ge k \big) \le v(X_0)\rho(X_0) \alpha^k.
\end{equation}
\end{theorem}

\begin{proof}
First, observe that the drift condition \eqref{eq:mult_beta_drift} extends to the entire switched domain $\widetilde{\mathcal{X}}$. For any state $X_t$ in the absorbing set $\widetilde{\mathcal{X}} \setminus \mathcal{X}$, the dynamics dictate $\widetilde{X}_{t+1} = X_t$. Since $\mathcal{T} \subseteq \mathcal{X}$, we have $\widetilde{X}_{t+1} \notin \mathcal{T}$, implying $\widetilde{v}(\widetilde{X}_{t+1}) = \alpha v(X_t)$. Thus, the condition $\mathbb{E}[\widetilde{v}(\widetilde{X}_{t+1}) \mid \widetilde{X}_t] = \alpha v(\widetilde{X}_t) \le \alpha v(\widetilde{X}_t) + \beta, \forall \widetilde{X}_t \in \widetilde{\mathcal{X}}$ holds trivially given $\beta \ge 0$.

\noindent\textbf{Step 1: Reformulate the drift condition.}
Define the state-dependent multiplier $\rho(x) := \alpha^{-1}$ if $x \in \mathcal{T}$, and $1$ otherwise.
The term $\widetilde{v}(\widetilde{X}_{t+1})$ can be factored as $\alpha \rho(\widetilde{X}_{t+1}) v(\widetilde{X}_{t+1})$. Dividing the barrier inequality by $\alpha$ yields the equivalent condition, we obtain
\begin{equation}
\label{eq:reformulated_drift}
\mathbb{E}\left[ \rho(\widetilde{X}_{t+1}) v(\widetilde{X}_{t+1}) \mid \widetilde{X}_t \right] \le v(\widetilde{X}_t) + \frac{\beta}{\alpha}, \quad \forall \widetilde{X}_t \in \widetilde{\mathcal{X}}.
\end{equation}

\noindent\textbf{Step 2: Define a scaled energy process.}
Define the stochastic process $\{Z_t\}_{t \ge 0}$ as $Z_t := v(\widetilde{X}_t) \prod_{s=0}^t \rho(\widetilde{X}_s)$.
Using the identity $\rho(x) = \alpha^{-\mathbf{1}_{\mathcal{T}}(x)}$, we have 
$\prod_{s=0}^t \rho(\widetilde{X}_s) = \alpha^{-\sum_{s=0}^t \mathbf{1}_{\mathcal{T}}(\widetilde{X}_s)} = \alpha^{-\widetilde{N}_{\mathcal{T}}(t,\omega)}$. Thus, $Z_t = \alpha^{-\widetilde{N}_{\mathcal{T}}(t,\omega)} v(\widetilde{X}_t)$. We now analyze the expected evolution of $Z_t$.
\begin{align*}
\mathbb{E}[Z_{t+1} \mid \mathcal{F}_t]= \left( \prod_{s=0}^t \rho(\widetilde{X}_s) \right) \mathbb{E}\left[ \rho(\widetilde{X}_{t+1}) v(\widetilde{X}_{t+1}) \mid \widetilde{X}_t \right].
\end{align*}
Applying \eqref{eq:reformulated_drift}, we have $\mathbb{E}[Z_{t+1} \mid \mathcal{F}_t] \le \alpha^{-\widetilde{N}_{\mathcal{T}}(t,\omega)} \left( v(\widetilde{X}_t) + \frac{\beta}{\alpha} \right) = Z_t + \frac{\beta}{\alpha} \alpha^{-\widetilde{N}_{\mathcal{T}}(t,\omega)}$.

\noindent\textbf{Step 3: Supermartingale bound.}
Define the process $M_t = Z_t - \sum_{j=0}^{t-1} D_j$, where $D_j = \frac{\beta}{\alpha} \alpha^{-\widetilde{N}_{\mathcal{T}}(j,\omega)}$. From the inequality above, $\mathbb{E}[M_{t+1} \mid \mathcal{F}_t] \le M_t$, so $M_t$ is a supermartingale.
We apply the Optional Stopping Theorem to the bounded stopping time $\tau = N \wedge \tau_k$, where \[\tau_k = \inf \{ t \ge 0 : \widetilde{N}_{\mathcal{T}}(t,\omega) = k \},\] to obtain $\mathbb{E}[Z_\tau] \le \mathbb{E}[Z_0] + \mathbb{E}\left[ \sum_{j=0}^{\tau-1} \frac{\beta}{\alpha} \alpha^{-\widetilde{N}_{\mathcal{T}}(j,\omega)} \right]$. To obtain a uniform upper bound, we maximize the drift term. The occupation count satisfies $\widetilde{N}_{\mathcal{T}}(j,\omega) \le j+1$ (at most 1 visit per step). Since $\alpha < 1$, the function $\alpha^{-x}$ is increasing, so $\alpha^{-\widetilde{N}_{\mathcal{T}}(j,\omega)} \le \alpha^{-(j+1)}$. Additionally, $\tau \le N$. Thus, $\sum_{j=0}^{\tau-1} \frac{\beta}{\alpha} \alpha^{-\widetilde{N}_{\mathcal{T}}(j,\omega)} \le \sum_{j=0}^{N-1} \frac{\beta}{\alpha} \alpha^{-(j+1)} = \frac{\beta}{\alpha} \sum_{t=1}^{N} \alpha^{-t}$ holds. Substituting this into the expectation (and noting $Z_0 = v(X_0)\rho(X_0)$), we have $\mathbb{E}[Z_\tau] \le v(X_0)\rho(X_0) + \frac{\beta}{\alpha} \sum_{t=1}^{N} \alpha^{-t}$.

\noindent\textbf{Step 4: Markov's Inequality.}
Consider the event $E_k = \{ \widetilde{N}_{\mathcal{T}}(N,\omega) \ge k \}$. On this event, the stopping time satisfies $\tau_k \le N$, so $\tau = \tau_k$. At this instant, the cumulative visit count is exactly $k$, i.e., $\widetilde{N}_{\mathcal{T}}(\tau_k,\omega) = k$. Furthermore, since the counter incremented at this step, we must have $\widetilde{X}_{\tau_k} \in \mathcal{T}$. By condition (2), $v(\widetilde{X}_{\tau_k}) \ge 1$. Thus, the value of the process is lower-bounded on this event:
\[
Z_\tau = \alpha^{-\widetilde{N}_{\mathcal{T}}(\tau_k,\omega)} v(\widetilde{X}_{\tau_k}) = \alpha^{-k} v(\widetilde{X}_{\tau_k}) \ge \alpha^{-k}.
\]
Since $Z_\tau$ is a non-negative random variable (as $v \ge 0$), we can apply Markov's inequality to obtain $\mathbb{P}(E_k) \le \mathbb{P}(Z_\tau \ge \alpha^{-k}) \le \frac{\mathbb{E}[Z_\tau]}{\alpha^{-k}} = \alpha^k \mathbb{E}[Z_\tau]$.
Substituting the upper bound for $\mathbb{E}[Z_\tau]$ derived in Step 3 yields the result for finite $N$.

\noindent\textbf{Step 5: Infinite Horizon Case ($\beta=0$).}
If $\beta=0$, the drift terms vanish, and $Z_t$ is a non-negative supermartingale. By the Optional Stopping Theorem for non-negative supermartingales (valid for unbounded $\tau_k$), we have
$\mathbb{E}[Z_{\tau_k} \mathbf{1}_{\{\tau_k < \infty\}}] \le \mathbb{E}[Z_0]$. On the event $\{\tau_k < \infty\}$ (equivalent to reaching count $k$ eventually), $Z_{\tau_k} \ge \alpha^{-k}$.
\[
\alpha^{-k} \mathbb{P}(\tau_k < \infty) \le v(X_0)\rho(X_0) \implies \mathbb{P}(\widetilde{N}_{\mathcal{T}}(\infty,\omega) \ge k) \le v(X_0)\rho(X_0)\alpha^k.
\]
This completes the proof. \qed
\end{proof}

In Theorem \ref{thm:mult_beta_occupation}, the additive drift parameter $\beta$ critically governs the asymptotic behavior of the upper bound. Since the contraction rate satisfies $\alpha \in (0,1)$, the cumulative drift term $\sum_{t=1}^{N} \alpha^{-t}$ grows geometrically with the horizon $N$. Specifically, the bound contains a term scaling as $\frac{\beta}{\alpha} \sum_{t=1}^{N} \alpha^{-t} \cdot \alpha^k \approx \mathcal{O}\left(\beta \cdot \alpha^{-(N-k)}\right)$. Consequently, any strictly positive drift $\beta > 0$ causes the bound to diverge as $N \to \infty$, reflecting the physical intuition that persistent noise in an unstable system may eventually drive the state into the target region given infinite time. Therefore, for infinite-horizon verification ($N=\infty$), Theorem \ref{thm:mult_beta_occupation} requires $\beta = 0$, in which case the bound simplifies to the strictly contractive form $v(x_0)\rho(x_0)\alpha^k$. For finite horizons, $\beta > 0$ is permissible, providing meaningful bounds that degrade gracefully as the horizon extends.

\subsection{Lower Bounds via Attractive Barriers}
\label{sub:sclb1}

In this subsection, we establish sufficient conditions to lower-bound the probability that the constrained occupation time exceeds a threshold $k$.

Theorem \ref{thm:attractive_mult_occupation_bound} presents a lower bound on the probability that the switched system \eqref{sdss1} visits the target $\mathcal{T}$ at least $k$ times within a finite horizon $N$. By Lemma \ref{lema1}, this corresponds directly to the probability that the original system \eqref{system} achieves at least $k$ visits before violating safety constraints. The core mechanism relies on an \textit{attractive barrier} $v(x)$, which quantifies the system's ``potential'' to visit the target. The barrier condition \eqref{eq:lower_drift_active_final} ensures that this potential is amplified by a factor $\alpha > 1$ upon every visit to $\mathcal{T}$, while strictly penalizing the system if it enters an absorbing sink state. The proof formalizes this dynamic by constructing a scaled process $Z_t = \alpha^{-(t-\widetilde{N}_{\mathcal{T}}(t,\omega))}v(X_t)$. Here, the cumulative visit count $\widetilde{N}_{\mathcal{T}}(t,\omega)$ in the exponent counteracts the natural time decay $\alpha^{-t}$. Consequently, frequent visits sustain the value of $Z_t$, whereas rare visits cause it to decay exponentially. This structure renders the (drift-adjusted) process a submartingale. By applying the Optional Stopping Theorem, we relate the initial potential $v(X_0)$ to the expected terminal state.

\begin{theorem}[Lower Bounds on Constrained Occupation Time]
\label{thm:attractive_mult_occupation_bound}
 Suppose there exist a measurable function $v: \widetilde{\mathcal{X}} \to \mathbb{R}$ and constants $\alpha > 1$ and $\beta \le 0$ such that $v$ is bounded on $\widetilde{\mathcal{X}}$, with $M := \sup_{X \in \widetilde{\mathcal{X}}} |v(X)| < \infty$, and satisfies the following conditions:

\begin{enumerate}
    \item \textbf{Attractive Barrier Condition:} For all $X_t \in \mathcal{X}$:
    \begin{equation}
    \label{eq:lower_drift_active_final}
    \mathbb{E}\left[ \alpha \mathbf{1}_{\mathcal{T}}(X_{t+1}) v(X_{t+1}) + \mathbf{1}_{\widetilde{\mathcal{X}} \setminus \mathcal{T}}(X_{t+1}) v(X_{t+1}) \;\middle|\; X_t \right] \ge \alpha v(X_t) + \beta.
    \end{equation}
    \item \textbf{Bound on the target:} \qquad $v(x) \le 1, \qquad \forall x \in \mathcal{T}$.
    \item \textbf{Sink Condition:} \qquad $v(x) \le -\frac{\beta}{\alpha - 1}, \qquad \forall x \in \widetilde{\mathcal{X}} \setminus \mathcal{X}$.
\end{enumerate}

Let $\rho(X) = \alpha$ if $X \in \mathcal{T}$ and $\rho(X) = 1$ otherwise.
For any initial condition $X_0 \in \mathcal{X}$, horizon $N \ge 1$, and target threshold $k \ge 1$, assume $k - \mathbf{1}_{\mathcal{T}}(X_0) \le N < \infty$. The probability that the constrained occupation time $\widetilde{N}_{\mathcal{T}}(N,\omega) \ge k$ satisfies
\begin{equation}
\label{eq:occupation_lower_bound_final}
\mathbb{P}\left( \widetilde{N}_{\mathcal{T}}(N,\omega) \ge k \right)
\;\ge\;
\frac{
v(X_0)\rho(X_0) + \Lambda(\beta, N, k') - M \alpha^{-(N - k + 1)}
}{
\rho(X_0) - M \alpha^{-(N - k + 1)}
},
\end{equation}
provided $\rho(X_0) > M \alpha^{-(N - k + 1)}$ and $X_0 \in \mathcal{X}$, where $\Lambda(\beta, N, k')=\beta N$.
\end{theorem}
\begin{proof}
Combining \eqref{eq:lower_drift_active_final} and the sink condition, we have \[\mathbb{E}\left[ \rho(\widetilde{X}_{t+1}) v(\widetilde{X}_{t+1}) \;\middle|\; \widetilde{X}_t \right] \ge \alpha v(\widetilde{X}_t) + \beta, \forall \widetilde{X}_t \in \widetilde{\mathcal{X}}.\]

Define the scaled process $\{Z_t\}_{t \ge 0}$ adapted to $\mathcal{F}_t$:
\[
Z_t := \alpha^{-t} \left( \prod_{s=0}^{t} \rho(\widetilde{X}_s) \right) v(\widetilde{X}_t)
= \alpha^{-(t - \widetilde{N}_{\mathcal{T}}(t,\omega))} v(\widetilde{X}_t).
\]

\noindent\textbf{Step 1: Submartingale Construction.}
Compute the conditional expectation:
\[
\begin{aligned}
\mathbb{E}[Z_{t+1} \mid \mathcal{F}_t] &= \alpha^{-(t+1)} \left(\prod_{s=0}^t \rho(\widetilde{X}_s)\right) \mathbb{E}[\rho(\widetilde{X}_{t+1}) v(\widetilde{X}_{t+1}) \mid \widetilde{X}_t] \\
&\ge \alpha^{-(t+1)} \alpha^{\widetilde{N}_{\mathcal{T}}(t,\omega)} (\alpha v(\widetilde{X}_t) + \beta)= Z_t + \beta \alpha^{-(t+1)} \alpha^{\widetilde{N}_{\mathcal{T}}(t,\omega)}.
\end{aligned}
\]
We now verify the submartingale property by bounding the drift term $D_t = \beta \alpha^{-(t+1)} \alpha^{\widetilde{N}_{\mathcal{T}}(t,\omega)}$. We must lower-bound the negative term. The maximum possible occupation count at time $t$ is $t+1$. Thus, $\alpha^{\widetilde{N}_{\mathcal{T}}(t,\omega)} \le \alpha^{t+1}$.
    Consequently, $\alpha^{-(t+1)} \alpha^{\widetilde{N}_{\mathcal{T}}(t,\omega)} \le \alpha^{-(t+1)} \alpha^{t+1} = 1$.  Multiplying by the non-positive constant $\beta$ reverses the inequality, we have \[D_t = \beta \left( \alpha^{-(t+1)} \alpha^{\widetilde{N}_{\mathcal{T}}(t,\omega)} \right) \ge \beta \cdot 1 = \beta.\]

Let $\Delta_t = \beta$. We have $\mathbb{E}[Z_{t+1} \mid \mathcal{F}_t] \ge Z_t + \Delta_t$.
The process $Y_t := Z_t - \sum_{j=0}^{t-1} \Delta_j$ is a submartingale.

\noindent\textbf{Step 2: Optional Stopping.}
Let $\tau_k := \inf \{t \ge 0 : \widetilde{N}_{\mathcal{T}}(t,\omega) = k\}$ and $\tau := \tau_k \wedge N$. Since $\tau \le N$ is bounded, we obtain from Proposition \ref{thm:OST} that
$\mathbb{E}[Z_\tau] \ge Z_0 + \mathbb{E}\left[ \sum_{j=0}^{\tau-1} \Delta_j \right]$.

We lower-bound the sum term $\Sigma_\tau = \sum_{j=0}^{\tau-1} \Delta_j$. Here $\Delta_j = \beta$. The sum is $\beta \tau$. Since $\beta$ is non-positive, we need an upper bound on $\tau$ to minimize the sum. Since $\tau \le N$, 
    $\Sigma_\tau = \beta \tau \ge \beta N$ holds.  Combining these cases into the definition of $\Lambda(\beta, N, k')$, we have $\mathbb{E}[Z_\tau] \ge v(X_0)\rho(X_0) + \Lambda(\beta, N, k')$.

\noindent\textbf{Step 3: Success--Failure Decomposition.}
Let $p := \mathbb{P}(\widetilde{N}_{\mathcal{T}}(N,\omega) \ge k)$.
\begin{itemize}
    \item \textbf{($\tau = \tau_k \le N$):}
    $Z_{\tau_k} \le \alpha^{\mathbf{1}_{\mathcal{T}}(X_0)} = \rho(X_0)$ (since $\tau_k\geq k'$ and $v(\widetilde{X}_{\tau_k})\leq 1$).
    
    \item \textbf{($\tau = N$):}
    $Z_N \le M \alpha^{-(N - k + 1)}$ (since $\widetilde{N}_{\mathcal{T}}(N,\omega)\leq k-1$ and $v(\widetilde{X})\leq M$, $\forall X \in \widetilde{\mathcal{X}}$).
\end{itemize}
Combining these, we have $\mathbb{E}[Z_\tau] \le p \cdot \rho(X_0) + (1 - p) M \alpha^{-(N - k + 1)}$.

\noindent\textbf{Step 4: Solve for $p$.}
Combining the lower and upper expectations, we obtain
\[
p \rho(X_0) + (1 - p) M \alpha^{-(N - k + 1)} \ge v(X_0)\rho(X_0) + \Lambda(\beta, N, k').
\]
Rearranging for $p$, we have $p \left( \rho(X_0) - M \alpha^{-(N - k + 1)} \right) \ge v(X_0)\rho(X_0) + \Lambda(\beta, N, k') - M \alpha^{-(N - k + 1)}$.
Dividing by $\rho(X_0) - M \alpha^{-(N - k + 1)}$ yields the result. \qed
\end{proof}

In Theorem \ref{thm:attractive_mult_occupation_bound}, we restrict our analysis to $\beta \le 0$. A positive drift $\beta > 0$ is structurally incompatible with the boundedness of the barrier function in this framework. Specifically, the attractive barrier condition with $\beta > 0$ would require the barrier value to grow indefinitely, contradicting the requirement that $v(x)$ is bounded (e.g., $v \le 1$ on the target). Therefore, no valid bounded barrier exists for $\beta > 0$ in this strict formulation. We must instead choose $\beta \le 0$, which introduces a trade-off: a larger negative drift simplifies finding a barrier but penalizes the probability bound as the horizon increases.
\begin{itemize}
    \item \textbf{Finite Horizons ($\beta < 0$):} The regime $\beta < 0$ is permissible for short-horizon tasks. It admits a broader class of candidate functions at the cost of a penalty term $\beta N$ that reduces the lower bound linearly with $N$.
    \item \textbf{Infinite Horizons ($\beta = 0$):} As $N \to \infty$, the penalty $\beta N$ diverges to $-\infty$ if $\beta < 0$, rendering the bound vacuous. Therefore, the strict submartingale regime $\beta = 0$ is essential and unique for long-horizon guarantees. It eliminates the additive drift entirely, relying solely on the geometric contraction $\alpha$ to certify repeated intersection.
\end{itemize}

We now extend the result to the infinite horizon ($N \to \infty$). Since the constrained occupation event is monotonic (if you visit $k$ times by horizon $N$, you also do so for any $N' > N$), we can derive the infinite-horizon bound as a limit of Theorem \ref{thm:attractive_mult_occupation_bound}. This yields a closed-form lower bound valid for the case $\beta = 0$.

\begin{corollary}[Infinite-Horizon Lower Bounds]
\label{cor:infinite_lower_bound}
Consider the setup of Theorem \ref{thm:attractive_mult_occupation_bound} with $\alpha > 1$. To obtain a non-trivial bound as $N \to \infty$, we must select $\beta = 0$. The probability of visiting the target set $\mathcal{T}$ at least $k$ times over the infinite horizon satisfies:
\begin{equation}
\label{eq:infinite_lower_bound}
\mathbb{P}\big( \widetilde{N}_{\mathcal{T}}(\infty,\omega) \ge k \big)
\;\ge\;
v(X_0).
\end{equation}
\end{corollary}
\begin{proof}
    The proof is shown in Appendix.
\end{proof}

In the infinite-horizon limit, the probability bound shifts from tracking a specific visit count, $k$, to measuring the likelihood of indefinite recurrence. Although the attraction factor $\alpha$ vanishes from the final expression, it remains the structural `engine' that dictates the barrier’s steepness and the strength of the guarantee. However, this reliance on strict geometric contraction renders the framework inherently conservative. The proposed `weighted' framework below resolves this by applying a stronger amplification ($\alpha^2$) within the target region. This modification creates the necessary mathematical 'headroom' to accommodate positive drift, ensuring the system remains attracted to the target despite the outward momentum

\begin{theorem}[Occupation-Time Lower Bounds via Weighted Attractive Barriers]
\label{thm:relaxed_occupation_bound}
Suppose there exists a measurable function $v: \widetilde{\mathcal{X}} \to \mathbb{R}$ and constants $\alpha > 1$, $\beta \ge 0$, such that $v$ is bounded on $\widetilde{\mathcal{X}}$, i.e., $M := \sup_{X \in \widetilde{\mathcal{X}}} |v(X)| < \infty$, and satisfies the following conditions:

\begin{enumerate}
    \item \textbf{Weighted Attractive Barrier Condition:} For all $X_t \in \mathcal{X}$:
    \begin{equation}
    \label{eq:relaxed_drift}
    \mathbb{E}\left[ \alpha^2 \mathbf{1}_{\mathcal{T}}(X_{t+1}) v(X_{t+1}) + \mathbf{1}_{\widetilde{\mathcal{X}} \setminus \mathcal{T}}(X_{t+1}) v(X_{t+1}) \;\middle|\; X_t \right] \ge \alpha v(X_t) + \beta.
    \end{equation}
    \item \textbf{Bound on the target:} \qquad $v(x) \le 1, \qquad \forall x \in \mathcal{T}$.
    \item \textbf{Sink Condition:} \qquad $v(x) \le -\frac{\beta}{\alpha - 1}, \qquad \forall x \in \widetilde{\mathcal{X}} \setminus \mathcal{X}$.
\end{enumerate}

Let the weighting function be $\rho(X) = \alpha^2$ if $X \in \mathcal{T}$ and $\rho(X) = 1$ otherwise.
For any initial condition $X_0 \in \mathcal{X}$, horizon $N \ge 1$, and target threshold $k \ge 1$, let $k' = k - \mathbf{1}_{\mathcal{T}}(X_0)$. Assume $k' \le N < \infty$. Define the accumulated drift term as $\Lambda(\beta, k') := \frac{\beta}{\alpha - 1}(1 - \alpha^{-k'})$. Then the probability that the constrained occupation time $\widetilde{N}_{\mathcal{T}}(N,\omega) \ge k$ satisfies
\begin{equation}
\label{eq:relaxed_lower_bound}
\mathbb{P}\left( \widetilde{N}_{\mathcal{T}}(N,\omega) \ge k \right)
\;\ge\;
\frac{
v(X_0)\rho(X_0) + \Lambda(\beta, k') - M \alpha^{-(N - 2k + 2)}
}{
\alpha^{k + \mathbf{1}_{\mathcal{T}}(X_0)} - M \alpha^{-(N - 2k + 2)}
},
\end{equation}
provided the denominator is positive.
\end{theorem}
\begin{proof}
    The proof is similar to that of Theorem \ref{thm:attractive_mult_occupation_bound}, which is shown in Appendix.
\end{proof}

\begin{corollary}[Infinite-Horizon Lower Bounds]
\label{cor:relaxed_infinite_bound} Consider the setup of Theorem \ref{thm:relaxed_occupation_bound} with $\alpha > 1$. The probability of visiting the target set $\mathcal{T}$ at least $k$ times over the infinite horizon satisfies
\begin{equation}
\label{eq:relaxed_infinite_bound}
\mathbb{P}\big( \widetilde{N}_{\mathcal{T}}(\infty,\omega) \ge k \big)
\;\ge\;
\frac{v(X_0)\rho(X_0) + \Lambda(\beta, k')}{\alpha^{k + \mathbf{1}_{\mathcal{T}}(X_0)}}.
\end{equation}
\end{corollary}
\begin{proof}
 The proof is the same as that of Corollary \ref{cor:infinite_lower_bound}. 
\end{proof}

The weighted condition introduces a fundamental trade-off. By allowing the barrier to decay inside the target (via the $\alpha^2$ weight), we enable the verification of systems with positive additive drift $\beta > 0$. However, this implies the system is ``leaky'' rather than strictly attractive. Thus, the probability bound in Corollary \ref{cor:relaxed_infinite_bound} decays geometrically with $k$ (scaling as $\alpha^{-k}$), reflecting that the likelihood of maintaining the repeated behavior may diminish with each subsequent visit.

\begin{remark}[Computational challenges and the indicator function]
While our theorems utilize barrier conditions that are affine in $v(x)$ for fixed $(\alpha,\beta)$, their computational synthesis presents challenges due to the indicator functions $\mathbf{1}_{\mathcal{T}}$ and $\mathbf{1}_{\mathcal{X}\setminus \mathcal{T}}$. These introduce piecewise constraints that must hold separately on different regions, rendering the synthesis problem non-convex and NP-hard for general continuous-state systems. However, these conditions remain valuable for guiding manual construction or tailoring computational methods in specific system classes. Furthermore, the inherent piecewise structure aligns well with emerging neural barrier function approaches (e.g., ReLU networks \cite{nair2010rectified}), offering principled loss functions for future research.
\end{remark}
\begin{remark}[Computational Tractability and Discrete Disturbances]
\label{remark2}
When the disturbance $d_t$ takes values from a finite discrete set $\{d_1, \dots, d_m\}$ with probabilities $\{p_1, \dots, p_m\}$, the expectation becomes a finite weighted sum. While the resulting expression is still piecewise due to the indicators, it can be rigorously verified by partitioning the safe set $\mathcal{X}$ into sub-regions where the target-hitting behavior of each disturbance mode is constant. For instance, consider $X_{t+1} = 0.5 X_t + d_t$ with $d_t \in \{-1, 1\}$ (uniform probability), $\mathcal{X} = [-4, 4]$, and $\mathcal{T} = [-0.5, 0.5]$. The pre-images of $\mathcal{T}$ for the two modes partition $\mathcal{X}$ into five intervals bounded by $\{-4, -3, -1, 1, 3, 4\}$. Within each interval, the indicator terms are constant. For example, in the region $[-3, -1]$, mode $d=-1$ maps outside $\mathcal{T}$ while mode $d=1$ maps inside. This reduces the verification to a standard polynomial inequality: $0.5 [\alpha v(0.5x - 1)] + 0.5 [1 \cdot v(0.5x + 1)] \le \alpha v(x) + \beta$. By checking such inequalities on each sub-interval (e.g., via SOS programming), one can rigorously certify the barrier without approximation errors.
\end{remark}

%% file: examples.tex
\section{Illustrating Examples}
\label{sec:examples}
In this section, we demonstrate the efficacy of the proposed barrier conditions using two representative scalar systems subject to discrete disturbances $d_t \in \{-0.1, 0.1\}$ with uniform probability. This setting allows for the exact computation of conditional expectations as finite sums. To rigorously handle the indicator function $\mathbf{1}_{\mathcal{T}}$ appearing in the barrier conditions, we partition the safe set $\mathcal{X}$ into sub-intervals based on whether the subsequent state lands within the target set $\mathcal{T}$ as discussed in Remark \ref{remark2}, and verify the required inequalities on each partition. We illustrate two complementary verification scenarios: (1) establishing upper bounds on the probability of visiting a remote target, and (2) establishing lower bounds on the probability of frequently visiting a central target.

\begin{example}[Upper-Bound Verification (Remote Target)]
\label{example1}
Consider a stable system where the target is located far from the equilibrium. Our objective is to certify that the probability of repeatedly visiting this remote region is negligible.

Consider the system $X_{t+1} = 0.5\,X_t + d_t$. The safe set is $\mathcal{X} = [-4, 4]$, the target set is $\mathcal{T} = [2, 3]$, and the initial state is $X_0 = 0$. We present two barrier functions to illustrate the trade-off between horizon specificity and horizon independence.

\noindent\textbf{Barrier 1 (Finite-horizon).} We employ a dissipative barrier $v_1(x) = (x/2)^4$ with a contraction rate of $\alpha_1 = 0.9$. Numerical verification confirms that the multiplicative barrier condition holds with a small additive drift $\beta_1 = 0.0002$.

Applying Theorem \ref{thm:mult_beta_occupation} with $\rho(X_0)=1$ (since $X_0 \notin \mathcal{T}$), we obtain the finite-horizon upper bound $\mathbb{P}(\widetilde{N}_{\mathcal{T}}(N,\omega) \ge k) \le \left(v_1(0) + \frac{\beta_1}{\alpha_1}\sum_{t=1}^{N}\alpha_1^{-t}\right)\alpha_1^k$.
The computed bounds for Barrier 1 are presented in Table \ref{tab:upper_results_finite}.

\begin{table}[h!]
\centering
\caption{Certified upper bounds for Example 1 using Barrier 1 (finite-horizon).}
\begin{tabular}{c c c c c}
\hline
\multirow{2}{*}{\textbf{Horizon} $N$} & \multicolumn{4}{c}{\textbf{Minimum Visits} ($k$)} \\
\cline{2-5}
 & $\mathbf{k=1}$ & $\mathbf{k=3}$ & $\mathbf{k=5}$ & $\mathbf{k=7}$ \\
\hline
10 & 0.0037 & 0.0030 & 0.0025 & 0.0020 \\
20 & 0.0146 & 0.0118 & 0.0096 & 0.0078 \\
30 & 0.0452 & 0.0366 & 0.0297 & 0.0240 \\
\hline
\end{tabular}
\label{tab:upper_results_finite}
\end{table}

\noindent\textbf{Barrier 2 (Infinite-horizon).} To obtain a horizon-independent bound, we must employ a barrier satisfying $\beta_2=0$. This requires the barrier to vanish on the system's robust invariant set $[-0.2, 0.2]$. We construct a dead-zone barrier $v_2(x) = 0.33 \max(0, |x| - 0.25)^2$. Numerical verification confirms this satisfies the barrier condition with $\alpha_2 = 0.3$ and $\beta_2 = 0$. Since the initial state $X_0=0$ lies within the dead-zone, $v_2(0)=0$. Applying Theorem \ref{thm:mult_beta_occupation} for $N=\infty$, we obtain $\mathbb{P}(\widetilde{N}_{\mathcal{T}}(\infty,\omega) \ge k) \le v_2(0)\rho(0)\alpha_2^k = 0$.
This barrier rigorously certifies that the probability of visiting the remote target is \textbf{zero} over an infinite horizon.

\noindent\textbf{Discussion.}
Barrier 1 uses a polynomial barrier with $\beta > 0$, providing valid but conservative bounds that increase with $N$. In contrast, Barrier 2 employs a dead-zone barrier with $\beta = 0$ that vanishes on the invariant set $[-0.2, 0.2]$, yielding a zero asymptotic bound. Both constructions are mathematically correct, but the differing bounds arise from the fundamental differences in their barrier designs and associated conditions. Barrier 1's conservatism grows with the horizon length, which can overshadow the system's actual contracting behavior, whereas Barrier 2 captures this behavior more precisely. This comparison highlights a key trade-off in barrier certificate design: stricter barrier conditions can provide tighter bounds but are harder to satisfy, while more permissive conditions offer valid but potentially conservative guarantees.
\end{example}

\begin{example}[Lower-Bound Verification (Central Target, Starting Outside)]
\label{example2}

Consider the same stable stochastic system, now initialized outside the target set. We show that a carefully designed piecewise-constant attractive barrier yields a non-trivial lower bound on the probability of frequent target visits.

We consider the system $X_{t+1} = 0.5\,X_t + d_t$. The safe set is $\mathcal{X} = [-1,1]$, the target set is $\mathcal{T} = [-0.2,0.2]$, and the initial state is $X_0 = 0.5$.

\noindent\textbf{Part 1: Attractive Barriers in Theorem \ref{thm:attractive_mult_occupation_bound} and Corollary \ref{cor:infinite_lower_bound}}

We use the attractive barrier function:
\[
v(x) =
\begin{cases}
1, & |x| \le 0.2, \\
0.99, & 0.2 < |x| \le 0.6, \\
0, & |x| > 0.6,
\end{cases}
\]
with amplification factor $\alpha = 1.009$ and drift $\beta = 0$.

\noindent \textbf{Verification of the Barrier Condition:}
We verify that, for all $x \in \mathcal{X}$,
\[
\mathbb{E}\!\left[
    \alpha \mathbf{1}_{\mathcal{T}}(X_{t+1}) v(X_{t+1})
    + \mathbf{1}_{\mathcal{X}\setminus\mathcal{T}}(X_{t+1}) v(X_{t+1})
    \,\middle|\, X_t = x
\right]
\;\ge\; \alpha\, v(x).
\]

1). If $x \in \mathcal{T}$, the target is invariant under the dynamics, and the condition holds with equality.

2). If $0.2 < |x| \le 0.6$, then $v(x)=0.99$. In the worst case, one disturbance realisation drives the state into $\mathcal{T}$ (yielding factor $\alpha$) while the other keeps it in $\mathcal{X}\setminus\mathcal{T}$ (yielding factor $1$). The condition reduces to $\tfrac{1}{2}\bigl(0.99 + \alpha\bigr)\ge \alpha \cdot 0.99$, which is satisfied for $\alpha = 1.009$ because $0.9995 \ge 0.9989$.
    
3). If $|x| > 0.6$, then $v(x)=0$, and the left-hand side is non-negative, so the inequality holds trivially.

The sink condition holds with global bound $M = \sup_{x \in \widetilde{\mathcal{X}}} |v(x)| = 1$. Applying Theorem \ref{thm:attractive_mult_occupation_bound} with $\rho(X_0)=1$ (since $X_0 \notin \mathcal{T}$), $k' = k$, and $\beta = 0$, we obtain the lower bound
\[
p \;\ge\;
\frac{v(0.5) - \alpha^{-(N-k+1)}}{1 - \alpha^{-(N-k+1)}}
\;=\;
\frac{0.99 - 1.009^{-(N-k+1)}}{1 - 1.009^{-(N-k+1)}}.
\]

Numerical results for varying $N$ and $k$ are shown in Table \ref{tab:comparison_bounds}. As $N \to \infty$, the bounds converge to $0.99$, certifying repeated interaction with the target.
\vspace{-0.0cm}
\begin{table}[ht]
\centering
\caption{Comparison of certified lower bounds obtained in Part 1 and Part 2.}
\label{tab:comparison_bounds}
\setlength{\tabcolsep}{3.5pt}
\begin{tabular}{c|cc|cc|cc}
\hline
\multirow{2}{*}{\textbf{Horizon} $N$} & \multicolumn{2}{c|}{\textbf{$k=5$}} & \multicolumn{2}{c|}{\textbf{$k=10$}} & \multicolumn{2}{c}{\textbf{$k=15$}} \\
\cline{2-7}
 & \textbf{Part 1} & \textbf{Part 2} & \textbf{Part 1} & \textbf{Part 2} & \textbf{Part 1} & \textbf{Part 2} \\
\hline
20       & 0.93 & 0.636 & 0.89 & 0.069 & 0.81 & 0.000 \\
50       & 0.97 & 0.846 & 0.97 & 0.696 & 0.96 & 0.525 \\
100      & 0.98 & 0.909 & 0.98 & 0.831 & 0.98 & 0.752 \\
200      & 0.99 & 0.936 & 0.99 & 0.885 & 0.99 & 0.834 \\
500      & 0.99 & 0.946 & 0.99 & 0.904 & 0.99 & 0.864 \\
$\infty$ & \textbf{0.99} & \textbf{0.947} & \textbf{0.99} & \textbf{0.905} & \textbf{0.99} & \textbf{0.866} \\
\hline
\end{tabular}
\end{table}
\vspace{-0.0cm}

\noindent\textbf{Part 2: Attractive Barriers in Theorem \ref{thm:relaxed_occupation_bound} and Corollary \ref{cor:relaxed_infinite_bound}}

We now apply the attractive barrier condition from Theorem \ref{thm:relaxed_occupation_bound} to the system. We utilize the same barrier function $v(x)$ and parameters ($\alpha = 1.009$, $\beta = 0$).

\noindent \textbf{Verification of the Barrier Condition:}
Theorem \ref{thm:relaxed_occupation_bound} requires a stronger weight inside the target ($\alpha^2$) to accommodate dynamics. We verify condition \eqref{eq:relaxed_drift}:
\[
\mathbb{E}\!\left[
    \alpha^2 \mathbf{1}_{\mathcal{T}}(X_{t+1}) v(X_{t+1})
    + \mathbf{1}_{\mathcal{X}\setminus\mathcal{T}}(X_{t+1}) v(X_{t+1})
    \,\middle|\, X_t = x
\right]
\;\ge\; \alpha\, v(x) + \beta.
\]
For $x \in \mathcal{T}$, the next state $X_{t+1}$ remains in $\mathcal{T}$. The LHS becomes $\alpha^2 \cdot v(X_{t+1}) = \alpha^2$, while the RHS is $\alpha \cdot v(x) = \alpha$. Since $\alpha > 1$, $\alpha^2 > \alpha$, satisfying the condition. For $x \notin \mathcal{T}$, the verification mirrors Part 1. Numerical results for varying $N$ and $k$ are shown in Table \ref{tab:comparison_bounds}.

\noindent \textbf{Discussion.}
The two barrier conditions yield qualitatively different guarantees. The first attractive barrier (Theorem \ref{thm:attractive_mult_occupation_bound}) provides a constant asymptotic bound of 0.99 for all $k$, certifying repeated interaction with the target. However, this requires the stricter condition of geometric amplification within the target set. In contrast, the  attractive barrier (Theorem \ref{thm:relaxed_occupation_bound}) allows for $\alpha^2$ amplification within the target, making the condition easier to satisfy but resulting in bounds that decay geometrically with $k$. This reflects a fundamental trade-off: while Theorem \ref{thm:attractive_mult_occupation_bound} provides a stronger guarantee (constant lower bound) but requires a stricter condition (geometric amplification), Theorem \ref{thm:relaxed_occupation_bound} relaxes the condition (allowing $\alpha^2$ amplification in the target) at the cost of a weaker, decaying bound. This flexibility allows Theorem \ref{thm:relaxed_occupation_bound} to verify systems where the strict contraction of Theorem \ref{thm:attractive_mult_occupation_bound} is overly conservative or infeasible. These complementary results demonstrate the framework's versatility in delivering tailored probabilistic guarantees for constrained occupation times, depending on the system characteristics and verification objectives. A comprehensive comparison of these two attractive barrier conditions will be investigated in the future work.

\end{example}

\noindent\textbf{Monte Carlo Validation.} To assess the validity and conservatism of the derived bounds, we performed Monte Carlo simulations with $10^5$ sample trajectories for both examples; representative trajectories are shown in Fig.~\ref{fig:simulation_examples}. For Example~\ref{example1}, the empirical probability of satisfying the visit-count condition was consistently $0.0$, in agreement with the certified upper bounds (e.g., $\le 2.4\%$ for $N=30$, $k=7$). This confirms that the barrier correctly identifies the target as effectively unreachable. For Example~\ref{example2}, the system exhibits strong performance: the empirical success probability reaches $1.0$ for all cases with $N \ge 50$, and remains at least $0.985$ when $N=20$. In all cases, the empirical results strictly exceed the certified lower bounds in Table~\ref{tab:comparison_bounds}. These results validate the proposed framework, demonstrating that the derived bounds provide rigorous guarantees on system behavior that remain valid beyond simulation-based evidence.

\begin{figure}
  \centering
    \includegraphics[width=0.8\textwidth,height=3.5cm]{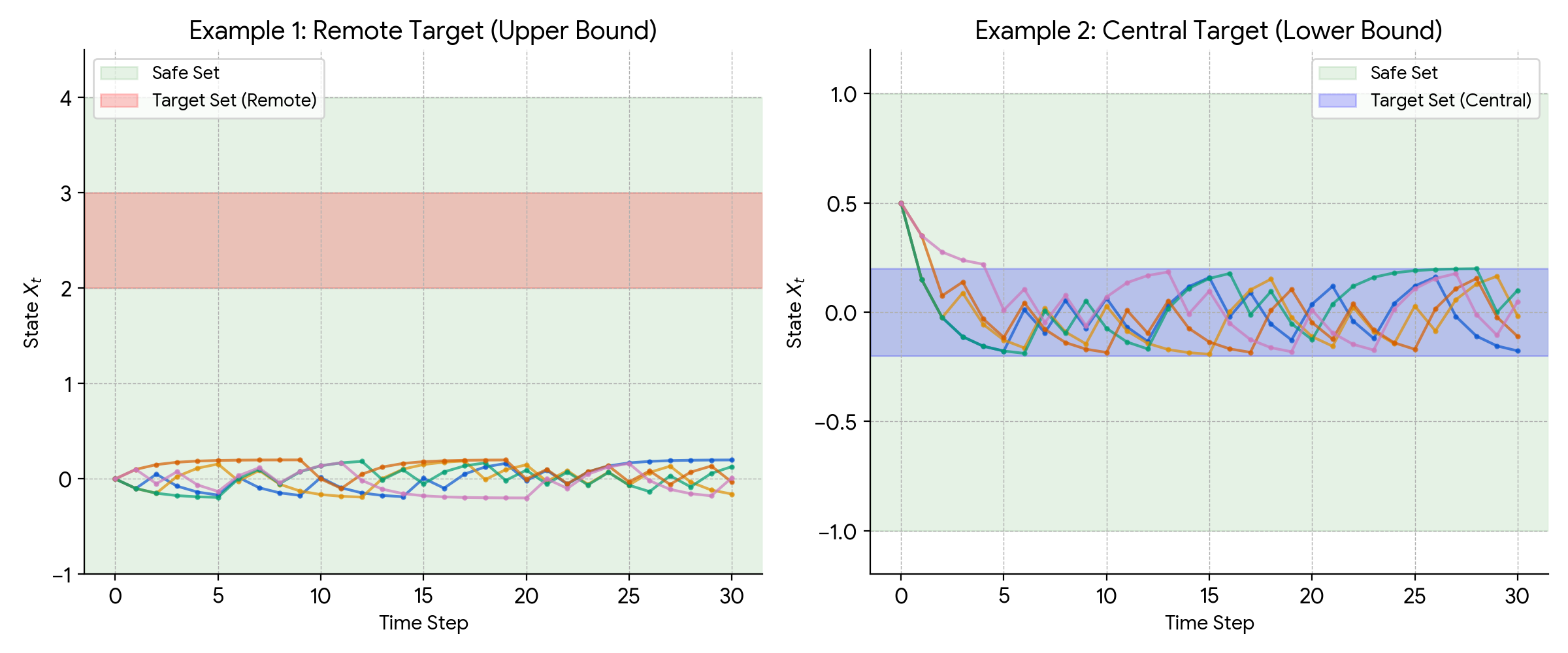}
        \captionsetup{justification=centering}
    \caption{Numerical illustration for constrained occupation time.} 
        \centering
    \label{fig:simulation_examples}
\end{figure}

%% file: conclusion.tex
\section{Conclusion}
\label{sec:con}
In this paper, we have presented a unified framework for the quantitative verification of constrained occupation time in stochastic discrete-time systems. By adopting a switched-system formulation and introducing dissipative and attractive stochastic barriers, we have extended classical reach-avoid analysis to capture the cumulative behaviors essential for repeated autonomy. While the resulting barrier conditions involve piecewise structures that challenge traditional convex synthesis, they establish the fundamental `Lyapunov-like' boundaries required for certification. Just as Foster-Lyapunov theory laid the groundwork for stability analysis long before efficient solvers existed, these conditions provide the rigorous theoretical target for future algorithmic development.

Notably, the piecewise-affine nature of our conditions makes them particularly well-suited for emerging Neural Barrier Function methods. Because ReLU networks naturally capture the required indicator geometries, our theoretical results directly suggest the appropriate loss functions and structural constraints for training neural certificates. This motivates a promising avenue for future investigation. In addition, data-driven computational methods equipped with PAC-style characterization constitute another important research direction. Further, future work will extend the proposed framework to continuous-time stochastic systems \cite{prajna2007framework,xue2024} and explore its application to the verification of general $\omega$-regular specifications \cite{henzinger2025supermartingale,abate2025quantitative}, by decomposing them into sequences of constrained occupation tasks.

%% file: appendix.tex
\section*{Appendix}
\setcounter{corollary}{0}
\begin{corollary}[Infinite-Horizon Lower Bounds]
Consider the setup of Theorem \ref{thm:attractive_mult_occupation_bound} with $\alpha > 1$. To obtain a non-trivial bound as $N \to \infty$, we must select $\beta = 0$. The probability of visiting the target set $\mathcal{T}$ at least $k$ times over the infinite horizon satisfies:
\begin{equation}
\label{eq:infinite_lower_bound}
\mathbb{P}\big( \widetilde{N}_{\mathcal{T}}(\infty,\omega) \ge k \big)
\;\ge\;
v(X_0).
\end{equation}
\end{corollary}
\begin{proof}
Let $E_j$ be the event $\{ \widetilde{N}_{\mathcal{T}}(j,\omega) \ge k \}$. Since the occupation count is cumulative, the sequence of events is nested: $E_j \subseteq E_{j+1}$. By the continuity of probability measures, we have $\mathbb{P}\big( \widetilde{N}_{\mathcal{T}}(\infty,\omega) \ge k \big) = \mathbb{P}\left( \bigcup_{j=1}^\infty E_j \right) = \lim_{j \to \infty} \mathbb{P}(E_j)$. From Theorem \ref{thm:attractive_mult_occupation_bound}, for any finite $j \ge k'$, we established the lower bound
\[
\mathbb{P}(E_j) \ge B(j) := \frac{ v(X_0)\rho(X_0) + \beta N - M \alpha^{-(j - k + 1)} }{ \rho(X_0) - M \alpha^{-(j - k + 1)} }.
\]
We now analyze the limit of $B(j)$ as $j \to \infty$.
Since $\alpha > 1$, the term $\alpha^{-(j - k + 1)}$ converges to $0$.
Provided the global bound $M$ is finite (as assumed in Theorem \ref{thm:attractive_mult_occupation_bound}), the term $M \alpha^{-(j - k + 1)}$ vanishes.
Thus, $\lim_{j \to \infty} B(N) = \frac{ v(X_0)\rho(X_0) + \beta N}{ \rho(X_0) }$.

We examine the limit as $N \to \infty$. If $\beta < 0$, the penalty term $\beta N$ diverges to $-\infty$, rendering the bound vacuous. Therefore, we require $\beta = 0$. In this case, the penalty term is 0, and the transient term $M \alpha^{-(N-k+1)} \to 0$ (since $\alpha > 1$). The bound simplifies directly to $v(X_0)$. \qed
\end{proof}

\setcounter{theorem}{2}

\begin{theorem}[Occupation-Time Lower Bounds via Weighted Attractive Barriers]
Suppose there exist a measurable function $v: \widetilde{\mathcal{X}} \to \mathbb{R}$ and constants $\alpha > 1$ and $\beta \ge 0$ such that $v$ is bounded on $\widetilde{\mathcal{X}}$, with $M := \sup_{X \in \widetilde{\mathcal{X}}} |v(X)| < \infty$. Assume the following conditions hold:

\begin{enumerate}
    \item \textbf{Weighted Attractive Barrier Condition:} For all $X_t \in \mathcal{X}$:
    \begin{equation}
    \label{eq:relaxed_drift1}
    \mathbb{E}\left[ \alpha^2 \mathbf{1}_{\mathcal{T}}(X_{t+1}) v(X_{t+1}) + \mathbf{1}_{\widetilde{\mathcal{X}} \setminus \mathcal{T}}(X_{t+1}) v(X_{t+1}) \;\middle|\; X_t \right] \ge \alpha v(X_t) + \beta.
    \end{equation}
    \item \textbf{Bound on the target:} \qquad $v(x) \le 1, \qquad \forall x \in \mathcal{T}$.
    \item \textbf{Sink Condition:} \qquad $v(x) \le -\frac{\beta}{\alpha - 1}, \qquad \forall x \in \widetilde{\mathcal{X}} \setminus \mathcal{X}$.
\end{enumerate}

Let the weighting function be $\rho(X) = \alpha^2$ if $X \in \mathcal{T}$ and $\rho(X) = 1$ otherwise.
For any initial condition $X_0 \in \mathcal{X}$, horizon $N \ge 1$, and target threshold $k \ge 1$, let $k' = k - \mathbf{1}_{\mathcal{T}}(X_0)$. Assume $k' \le N < \infty$. Define the accumulated drift term as $\Lambda(\beta, k') := \frac{\beta}{\alpha - 1}(1 - \alpha^{-k'})$. Then the probability that the constrained occupation time $\widetilde{N}_{\mathcal{T}}(N,\omega) \ge k$ satisfies:
\begin{equation*}
\label{eq:relaxed_lower_bound1}
\mathbb{P}\left( \widetilde{N}_{\mathcal{T}}(N,\omega) \ge k \right)
\;\ge\;
\frac{
v(X_0)\rho(X_0) + \Lambda(\beta, k') - M \alpha^{-(N - 2k + 2)}
}{
\alpha^{k + \mathbf{1}_{\mathcal{T}}(X_0)} - M \alpha^{-(N - 2k + 2)}
},
\end{equation*}
provided the denominator is positive.
\end{theorem}
\begin{proof}
The proof follows the martingale construction strategy but adapts the scaling factors to match the weighting factor $\alpha^2$.

\noindent\textbf{Preliminaries:}
The weighted barrier condition \eqref{eq:relaxed_drift1} can be equivalently stated as $\mathbb{E}[\rho(\widetilde{X}_{t+1}) v(\widetilde{X}_{t+1}) \mid \widetilde{X}_t] \ge \alpha v(\widetilde{X}_t) + \beta$ for all $\widetilde{X}_t \in \mathcal{X}$. For the sink states $\widetilde{X} \in \widetilde{\mathcal{X}} \setminus \mathcal{X}$, the Sink Condition (3) ensures the inequality holds globally.

\noindent\textbf{Step 1: Submartingale Construction.}
We define a scaled stochastic process $\{Z_t\}_{t \ge 0}$ adapted to $\mathcal{F}_t$:
\[
Z_t := \alpha^{-t} \left( \prod_{s=0}^{t} \rho(\widetilde{X}_s) \right) v(\widetilde{X}_t).
\]
Since the product term accumulates a factor of $\alpha^2$ upon every visit to $\mathcal{T}$, we can express $Z_t$ as:
\[
Z_t = \alpha^{-t} \alpha^{2\widetilde{N}_{\mathcal{T}}(t,\omega)} v(\widetilde{X}_t) = \alpha^{-(t - 2\widetilde{N}_{\mathcal{T}}(t,\omega))} v(\widetilde{X}_t).
\]
We verify the submartingale property by computing the conditional expectation:
\[
\begin{aligned}
\mathbb{E}[Z_{t+1} \mid \mathcal{F}_t] &= \alpha^{-(t+1)} \alpha^{2\widetilde{N}_{\mathcal{T}}(t,\omega)} \mathbb{E}[\rho(\widetilde{X}_{t+1}) v(\widetilde{X}_{t+1}) \mid \widetilde{X}_t] \\
&\ge \alpha^{-(t+1)} \alpha^{2\widetilde{N}_{\mathcal{T}}(t,\omega)} (\alpha v(\widetilde{X}_t) + \beta) \\
&= Z_t + \beta \alpha^{-(t+1)} \alpha^{2\widetilde{N}_{\mathcal{T}}(t,\omega)}.
\end{aligned}
\]
Since $\beta \ge 0$ and $\widetilde{N}_{\mathcal{T}} \ge 0$, the drift term $D_t \ge \beta \alpha^{-(t+1)}$. Let $\Delta_t = \beta \alpha^{-(t+1)}$. The process $Y_t = Z_t - \sum_{j=0}^{t-1} \Delta_j$ is a submartingale.

\noindent\textbf{Step 2: Optional Stopping Theorem.}
Let $\tau = \inf \{t \ge 0 : \widetilde{N}_{\mathcal{T}}(t,\omega) = k\} \wedge N$. By the Optional Stopping Theorem, $\mathbb{E}[Z_\tau] \ge Z_0 + \mathbb{E}[\sum_{j=0}^{\tau-1} \Delta_j]$. Using the fact that $\tau \ge k'$ (minimum steps to record $k'$ visits), we bound the sum by $\Lambda(\beta, k')$, yielding $\mathbb{E}[Z_\tau] \ge v(X_0)\rho(X_0) + \Lambda(\beta, k')$.

\noindent\textbf{Step 3: Success--Failure Decomposition.}
We decompose $\mathbb{E}[Z_\tau]$ based on the event of success $S = \{\widetilde{N}_{\mathcal{T}}(N) \ge k\}$:
\begin{itemize}
    \item \textbf{Success ($S$):} At $\tau$, the visit count is exactly $k$. Thus $Z_\tau = \alpha^{2k - \tau} v(\widetilde{X}_\tau)$. Since $\tau \ge k'$ and $v \le 1$, we have $Z_\tau \le \alpha^{2k - k'} = \alpha^{k + \mathbf{1}_{\mathcal{T}}(X_0)}$.
    \item \textbf{Failure ($S^c$):} At $\tau=N$, the visit count is $\le k-1$. Thus $Z_N \le M \alpha^{-(N - 2k + 2)}$.
\end{itemize}
Combining these bounds, we obtain
\[
p \cdot \alpha^{k + \mathbf{1}_{\mathcal{T}}(X_0)} + (1-p) M \alpha^{-(N - 2k + 2)} \ge v(X_0)\rho(X_0) + \Lambda(\beta, k').
\]
Solving for $p$ yields the result. \qed
\end{proof}

%% file: ref.bib
@article{wang2025quantitative,
  title={Quantitative Verification of Omega-regular Properties in Probabilistic Programming},
  author={Wang, Peixin and Bai, Jianhao and Zhang, Min and Ong, C-H Luke},
  journal={arXiv preprint arXiv:2512.21596},
  year={2025}
}

@article{abate2008probabilistic,
  title={Probabilistic reachability and safety for controlled discrete time stochastic hybrid systems},
  author={Abate, Alessandro and Prandini, Maria and Lygeros, John and Sastry, Shankar},
  journal={Automatica},
  volume={44},
  number={11},
  pages={2724--2734},
  year={2008},
  publisher={Elsevier}
}

@article{cao2025comparative,
  title={Comparative Analysis of Barrier-like Function Methods for Reach-Avoid Verification in Stochastic Discrete-Time Systems},
  author={Cao, Zhipeng and Wang, Peixin and Ong, Luke and {\v{Z}}ikeli{\'c}, {\DJ}or{\dj}e and Wagner, Dominik and Xue, Bai},
  journal={arXiv preprint arXiv:2512.05348},
  year={2025}
}

@article{xue2024finite,
  title={Finite-time safety and reach-avoid verification of stochastic discrete-time systems},
  author={Xue, Bai},
  journal={Information and Computation},
  volume = {307},
pages = {105368},
issn = {0890-5401},
  year={2025}
}

@article{xue2024sufficient,
  title={Sufficient and necessary barrier-like conditions for safety and reach-avoid verification of stochastic discrete-time systems},
  author={Xue, Bai},
  journal={Automatica},
  volume={187},
  pages={112919},
  year={2026},
  publisher={Elsevier}
}

@inproceedings{abate2025quantitative,
  title={Quantitative supermartingale certificates},
  author={Abate, Alessandro and Giacobbe, Mirco and Roy, Diptarko},
  booktitle={International Conference on Computer Aided Verification},
  pages={3--28},
  year={2025},
  organization={Springer}
}

@inproceedings{chatterjee2017stochastic,
  title={Stochastic invariants for probabilistic termination},
  author={Chatterjee, Krishnendu and Novotn{\`y}, Petr and {\v{Z}}ikeli{\'c}, {\DH}or{\dj}e},
  booktitle={Proceedings of the 44th ACM SIGPLAN Symposium on Principles of Programming Languages},
  pages={145--160},
  year={2017}
}

@article{summers2010verification,
  title={Verification of discrete time stochastic hybrid systems: A stochastic reach-avoid decision problem},
  author={Summers, Sean and Lygeros, John},
  journal={Automatica},
  volume={46},
  number={12},
  pages={1951--1961},
  year={2010},
  publisher={Elsevier}
}

@inproceedings{chakarov2013probabilistic,
  title={Probabilistic program analysis with martingales},
  author={Chakarov, Aleksandar and Sankaranarayanan, Sriram},
  booktitle={Computer Aided Verification: 25th International Conference, CAV 2013, Saint Petersburg, Russia, July 13-19, 2013. Proceedings 25},
  pages={511--526},
  year={2013},
  organization={Springer}
}

@inproceedings{jagtap2018temporal,
  title={Temporal logic verification of stochastic systems using barrier certificates},
  author={Jagtap, Pushpak and Soudjani, Sadegh and Zamani, Majid},
  booktitle={International Symposium on Automated Technology for Verification and Analysis},
  pages={177--193},
  year={2018},
  organization={Springer}
}

@inproceedings{yu2023safe,
  title={Safe probabilistic invariance verification for stochastic discrete-time dynamical systems},
  author={Yu, Yiqing and Wu, Taoran and Xia, Bican and Wang, Ji and Xue, Bai},
  booktitle={2023 62nd IEEE Conference on Decision and Control (CDC)},
  pages={5804--5811},
  year={2023},
  organization={IEEE}
}

@inproceedings{henzinger2025supermartingale,
  title={Supermartingale certificates for quantitative omega-regular verification and control},
  author={Henzinger, Thomas A and Mallik, Kaushik and Sadeghi, Pouya and {\v{Z}}ikeli{\'c}, {\DJ}or{\dj}e},
  booktitle={International Conference on Computer Aided Verification},
  pages={29--55},
  year={2025},
  organization={Springer}
}

@inproceedings{zhi2024unifying,
  title={Unifying qualitative and quantitative safety verification of DNN-controlled systems},
  author={Zhi, Dapeng and Wang, Peixin and Liu, Si and Ong, C-H Luke and Zhang, Min},
  booktitle={International Conference on Computer Aided Verification},
  pages={401--426},
  year={2024},
  organization={Springer}
}

@article{vzikelic2023compositional,
  title={Compositional policy learning in stochastic control systems with formal guarantees},
  author={{\v{Z}}ikeli{\'c}, {\DJ}or{\dj}e and Lechner, Mathias and Verma, Abhinav and Chatterjee, Krishnendu and Henzinger, Thomas},
  journal={Advances in Neural Information Processing Systems},
  volume={36},
  pages={47849--47873},
  year={2023}
}

@inproceedings{xue2021reach,
  title={Reach-avoid Analysis for Stochastic Discrete-time Systems},
  author={Xue, Bai and Li, Renjue and Zhan, Naijun and Fr{\"a}nzle, Martin},
  booktitle={2021 American Control Conference (ACC)},
  pages={4879--4885},
  year={2021},
  organization={IEEE}
}

@article{prajna2007framework,
  title={A framework for worst-case and stochastic safety verification using barrier certificates},
  author={Prajna, Stephen and Jadbabaie, Ali and Pappas, George J},
  journal={IEEE Transactions on Automatic Control},
  volume={52},
  number={8},
  pages={1415--1428},
  year={2007},
  publisher={IEEE}
}

@article{takisaka2021ranking,
  title={Ranking and repulsing supermartingales for reachability in randomized programs},
  author={Takisaka, Toru and Oyabu, Yuichiro and Urabe, Natsuki and Hasuo, Ichiro},
  journal={ACM Transactions on Programming Languages and Systems (TOPLAS)},
  volume={43},
  number={2},
  pages={1--46},
  year={2021},
  publisher={ACM New York, NY, USA}
}

@inproceedings{kenyon2021supermartingales,
  title={Supermartingales, ranking functions and probabilistic lambda calculus},
  author={Kenyon-Roberts, Andrew and Ong, C-H Luke},
  booktitle={2021 36th Annual ACM/IEEE Symposium on Logic in Computer Science (LICS)},
  pages={1--13},
  year={2021},
  organization={IEEE}
}

@book{williams1991probability,
  title={Probability with martingales},
  author={Williams, David},
  year={1991},
  publisher={Cambridge university press}
}

@book{grimmett2020probability,
  title={Probability and random processes},
  author={Grimmett, Geoffrey and Stirzaker, David},
  year={2020},
  publisher={Oxford university press}
}

@inproceedings{abate2024stochastic,
  title={Stochastic omega-regular verification and control with supermartingales},
  author={Abate, Alessandro and Giacobbe, Mirco and Roy, Diptarko},
  booktitle={International Conference on Computer Aided Verification},
  pages={395--419},
  year={2024},
  organization={Springer}
}

@incollection{gordon2014probabilistic,
  title={Probabilistic programming},
  author={Gordon, Andrew D and Henzinger, Thomas A and Nori, Aditya V and Rajamani, Sriram K},
  booktitle={Future of software engineering proceedings},
  pages={167--181},
  year={2014}
}

@article{godreche2001statistics,
  title={Statistics of the occupation time of renewal processes},
  author={Godreche, C and Luck, JM1853425},
  journal={Journal of Statistical Physics},
  volume={104},
  number={3},
  pages={489--524},
  year={2001},
  publisher={Springer}
}

@article{darling1957occupation,
  title={On occupation times for Markoff processes},
  author={Darling, Donald A and Kac, Mark},
  journal={Transactions of the American Mathematical Society},
  volume={84},
  number={2},
  pages={444--458},
  year={1957},
  publisher={JSTOR}
}

@inproceedings{nair2010rectified,
  title={Rectified linear units improve restricted boltzmann machines},
  author={Nair, Vinod and Hinton, Geoffrey E},
  booktitle={Proceedings of the 27th international conference on machine learning (ICML-10)},
  pages={807--814},
  year={2010}
}

@inproceedings{pnueli1977temporal,
  title={The temporal logic of programs},
  author={Pnueli, Amir},
  booktitle={18th annual symposium on foundations of computer science (sfcs 1977)},
  pages={46--57},
  year={1977},
  organization={ieee}
}

@book{soderstrom2012discrete,
  title={Discrete-time stochastic systems: estimation and control},
  author={S{\"o}derstr{\"o}m, Torsten},
  year={2012},
  publisher={Springer Science \& Business Media}
}

@inproceedings{clarke2011statistical,
  title={Statistical model checking for cyber-physical systems},
  author={Clarke, Edmund M and Zuliani, Paolo},
  booktitle={International symposium on automated technology for verification and analysis},
  pages={1--12},
  year={2011},
  organization={Springer}
}

@article{steinhardt2012finite,
  title={Finite-time regional verification of stochastic non-linear systems},
  author={Steinhardt, Jacob and Tedrake, Russ},
  journal={The International Journal of Robotics Research},
  volume={31},
  number={7},
  pages={901--923},
  year={2012},
  publisher={SAGE Publications Sage UK: London, England}
}

@article{prajna2007convex,
  title={Convex programs for temporal verification of nonlinear dynamical systems},
  author={Prajna, Stephen and Rantzer, Anders},
  journal={SIAM Journal on Control and Optimization},
  volume={46},
  number={3},
  pages={999--1021},
  year={2007},
  publisher={SIAM}
}

@inproceedings{clarke1997model,
  title={Model checking},
  author={Clarke, Edmund M},
  booktitle={Foundations of Software Technology and Theoretical Computer Science: 17th Conference Kharagpur, India, December 18--20, 1997 Proceedings 17},
  pages={54--56},
  year={1997},
  organization={Springer}
}

@book{manna2012temporal,
  title={Temporal verification of reactive systems: safety},
  author={Manna, Zohar and Pnueli, Amir},
  year={2012},
  publisher={Springer Science \& Business Media}
}

@book{baier2008principles,
  title={Principles of model checking},
  author={Baier, Christel and Katoen, Joost-Pieter},
  year={2008},
  publisher={MIT press}
}

@article{santoyo2021barrier,
  title={A barrier function approach to finite-time stochastic system verification and control},
  author={Santoyo, Cesar and Dutreix, Maxence and Coogan, Samuel},
  journal={Automatica},
  volume={125},
  pages={109439},
  year={2021},
  publisher={Elsevier}
}

@article{mathiesen2022safety,
  title={Safety certification for stochastic systems via neural barrier functions},
  author={Mathiesen, Frederik Baymler and Calvert, Simeon C and Laurenti, Luca},
  journal={IEEE Control Systems Letters},
  volume={7},
  pages={973--978},
  year={2022},
  publisher={IEEE}
}

@article{xue2024,
  author={Xue, Bai and Zhan, Naijun and Fränzle, Martin},
  journal={IEEE Transactions on Automatic Control}, 
  title={Reach-Avoid Analysis for Polynomial Stochastic Differential Equations}, 
  year={2024},
  volume={69},
  number={3},
  pages={1882-1889}
}

@article{kushner1967stochastic,
  title={Stochastic stability and control.},
  author={Kushner, Harold Joseph},
  year={1967},
  publisher={ACADEMIC PRESS, INC.}
}

@inproceedings{chakarov2016deductive,
  title={Deductive proofs of almost sure persistence and recurrence properties},
  author={Chakarov, Aleksandar and Voronin, Yuen-Lam and Sankaranarayanan, Sriram},
  booktitle={International Conference on Tools and Algorithms for the Construction and Analysis of Systems},
  pages={260--279},
  year={2016},
  organization={Springer}
}
